\documentclass[envcountsect,envcountsame]{llncs}

\usepackage{amssymb}
\usepackage{bussproofs}
\usepackage{listings}
\usepackage[all]{xy}
\SelectTips{cm}{} % - Looks pretty

\newcommand{\atms}{\mathrm{Atm}}
\newcommand{\condzero}{\mathrm{C0}}
\newcommand{\condone}{\mathrm{C1}}
\newcommand{\condtwo}{\mathrm{C2}}
\newcommand{\eqref}[1]{(\ref{#1})}
\newcommand{\forces}{\Vdash}

\newcommand{\iimp}{\twoheadrightarrow}
\newcommand{\implist}[5]{#2_{#4} #1 {#3}_{#4} , \cdots ,  #2_{#5}
  #1 {#3}_{#5}}
\newcommand{\lcnxt}{\mathrm{KM}_{lin}}

\newcommand{\lgint}{\mathrm{Int}}

\newcommand{\lgkm}{\mathrm{KM}}
\newcommand{\lglc}{\mathrm{LC}}

\newcommand{\limp}{\rightarrow}
\newcommand{\mprulename}{\small mp}
\newcommand{\nxt}{\rhd}
\newcommand{\later}{\nxt}
\newcommand{\nxtlist}[2]{\nxt {#1}_{1} , \cdots ,  \nxt {#1}_{#2}}
\newcommand{\oplist}[4]{#1 {#2}_{#3} , \cdots ,  #1 #2_{#4}}
\newcommand{\opset}[2]{{#2}^{#1}}
\newcommand{\opsetmi}[3]{\opset{#1}{#2}_{-#3}}
\newcommand{\rel}{R}
\newcommand{\releq}{R^{=}}
\newcommand{\res}[2]{r^{#1}_{#2}}
\newcommand{\seq}{\vdash}
\newcommand{\seqlcnxt}{\mathrm{SKM}_{lin}}
\newcommand{\seqkm}{\mathrm{SKM}}
\newcommand{\trees}{\mathcal{S}}
\newcommand{\wbx}{\Box}

\newcommand{\toprightrulename}{\small \top \mathrm{R}}
\newcommand{\toprule}{
\AxiomC{}
\LeftLabel{$\toprightrulename$}
\UnaryInfC{$\Gamma \seq \top, \Delta $}
\DisplayProof
}

\newcommand{\idrulename}{\small id}
\newcommand{\idrule}{
\AxiomC{}
\LeftLabel{\idrulename}
\UnaryInfC{$\Gamma, \varphi \seq \varphi, \Delta $}
\DisplayProof
}

\newcommand{\botleftrulename}{\small \bot \mathrm{L}}
\newcommand{\botrule}{
\AxiomC{}
\LeftLabel{$\botleftrulename$}
\UnaryInfC{$\Gamma , \bot \seq \Delta $}
\DisplayProof
}

\newcommand{\orleftrulename}{\small \lor \mathrm{L}}
\newcommand{\orleftrule}{
\AxiomC{$\Gamma , \varphi \seq \Delta $}
\AxiomC{$\Gamma , \psi \seq  \Delta $}
\LeftLabel{$\orleftrulename$}
\BinaryInfC{$\Gamma, \varphi \lor \psi \seq \Delta $}
\DisplayProof
}

\newcommand{\orrightrulename}{\small \lor \mathrm{R}}
\newcommand{\orrightrule}{
\AxiomC{$\Gamma \seq  \varphi, \psi , \Delta $}
\LeftLabel{$\orrightrulename$}
\UnaryInfC{$\Gamma  \seq \varphi\lor\psi , \Delta $}
\DisplayProof
}

\newcommand{\andrightrulename}{\small \land \mathrm{R}}
\newcommand{\andrightrule}{
\AxiomC{$\Gamma  \seq \varphi, \Delta $}
\AxiomC{$\Gamma  \seq \psi,  \Delta $}
\LeftLabel{$\andrightrulename$}
\BinaryInfC{$\Gamma \seq \varphi\land\psi , \Delta $}
\DisplayProof
}

\newcommand{\andleftrulename}{\small \land \mathrm{L}}
\newcommand{\andleftrule}{
\AxiomC{$\Gamma , \varphi , \psi \seq \Delta $}
\LeftLabel{$\andleftrulename$}
\UnaryInfC{$\Gamma , \varphi\land\psi \seq \Delta $}
\DisplayProof
}

\newcommand{\impleftrulename}{\small \limp\!\!\mathrm{L}}
\newcommand{\impleftrule}{
\AxiomC{$\Gamma  , \varphi\iimp\psi  \seq \varphi, \Delta $}
\AxiomC{$\Gamma , \varphi\iimp\psi, \psi \seq  \Delta $}
\LeftLabel{$\impleftrulename$}
\BinaryInfC{$\Gamma , \varphi\limp\psi \seq \Delta $}
\DisplayProof
}

\newcommand{\imprightrulename}{\small \limp\!\!\mathrm{R}}
\newcommand{\imprightrule}{
\AxiomC{$\Gamma , \varphi \seq \psi ,  \Delta$}
\AxiomC{$\Gamma  \seq \varphi\iimp\psi ,\Delta$}
\LeftLabel{$\imprightrulename$}
\BinaryInfC{$\Gamma \seq \varphi\limp\psi , \Delta$}
\DisplayProof
}

\newcommand{\steprulename}{\textsc{step}}
\newcommand{\newsteprulename}{\textsc{step}}

\newcommand{\newsteprulenewversion}{
   \begin{tabular}[c]{l@{\extracolsep{2cm}}l}
   \multicolumn{2}{c}{
   \AxiomC{$\mathrm{Prem}_1 
         \quad \cdots \quad 
         \mathrm{Prem}_k 
         \quad 
         \mathrm{Prem}_{k+1} 
         \quad \cdots
         \quad \mathrm{Prem}_{k+n}$}
   \LeftLabel{$\steprulename$}
   \RightLabel{$\dagger$}
   \UnaryInfC{$\Sigma_l, \opset{\nxt}{\Theta} ,\opset{\iimp}{\Gamma} 
         \seq 
        \opset{\iimp}{\Delta},
        \opset{\nxt}{\Phi},
         \Sigma_r$}
   \noLine
   \UnaryInfC{
     \begin{tabular}[c]{lll}
     \\%[5px]
      $\qquad\qquad
      \mathrm{Prem}_{1 \leq i \leq k}$ 
      & $=$
      & $\Sigma_l, 
     \Theta , 
     \opset{\nxt}{\Theta},
     \opset{\limp}{\Gamma} ,
     \varphi_i \iimp \psi_i , 
     \varphi_i 
     \seq
     \psi_i,
     \opsetmi{\limp}{\Delta}{i},
     \Phi$
   \\[5px]
   $\qquad\qquad
    \mathrm{Prem}_{k+1 \leq i \leq k+n}$
   & $=$
   & $\Sigma_l, 
    \Theta, 
     \opset{\nxt}{\Theta},
    \opset{\limp}{\Gamma},
    \nxt \phi_{i-k} 
    \seq
    \opset{\limp}{\Delta},
    \Phi$
   \end{tabular}
}
\DisplayProof
}
 \\%[5px]
 \\%[5px]
  $\opset{\nxt}{\Theta} = \nxtlist{\theta}{j}$
    &
  $\opset{}{\Theta} = \oplist{}{\theta}{1}{j}$
   \\[2px]
  $\opset{\iimp}{\Gamma} = \{\implist{\iimp}{\alpha}{\beta}{1}{l}\}$
      &
  $\opset{\limp}{\Gamma} = \{\implist{\limp}{\alpha}{\beta}{1}{l}\}$
   \\[2px]
  $\opset{\iimp}{\Delta} = \{\implist{\iimp}{\varphi}{\psi}{1}{k}\}$
      &
  $\opset{\limp}{\Delta} = \{\implist{\limp}{\varphi}{\psi}{1}{k}\}$
\\[2px]
  $\opsetmi{\limp}{\Delta}{i} = 
     \opset{\limp}{\Delta}\setminus 
     \{\varphi_i\limp\psi_i\}
  $
   \\[2px]
  $\opset{\nxt}{\Phi} = \nxt\phi_1 , \cdots , \nxt\phi_n$
    &
  $\opset{}{\Phi} = \phi_1 , \cdots , \phi_n $
   \\[5px]
   \multicolumn{2}{l}{
   where $\dagger$ means that the conditions $\condzero$, $\condone$ 
   and $\condtwo$ below must hold
  }
\\[2px]
   \multicolumn{2}{l}{
($\condzero$)   $\opset{\iimp}{\Delta} \cup \opset{\nxt}{\Phi}  \neq \emptyset$
  }
\\[2px]
   \multicolumn{2}{l}{
($\condone$)    
    $\bot \not\in\Sigma_l$ and $\top\not\in\Sigma_r$ and 
    $(\Sigma_l \cup \opset\nxt\Theta \cup \Gamma^\iimp ) \cap 
       (\opset{\iimp}{\Delta} 
        \cup 
        \opset{\nxt}{\Phi} 
        \cup 
        \Sigma_r)  = \emptyset$
   }
   \\[2px]
   \multicolumn{2}{l}{
($\condtwo$) 
    $\Sigma_l$ and $\Sigma_r$ each contain atomic formulae only 
   }
  \\[5px]
   \multicolumn{2}{l}{
   Explanations for the conditions:}
   \\[2px]
   \multicolumn{2}{l}{
   ($\condzero$) there must be at least one $\nxt$- or $\iimp$-formula in the
   succedent of the conclusion}
   \\[2px]
   \multicolumn{2}{l}{
   ($\condone$) none of the rules
       $\botleftrulename, \toprightrulename, \idrulename$ 
       are applicable to the conclusion}
   \\[2px]
   \multicolumn{2}{l}{
   ($\condtwo$) none of the rules $\orleftrulename,, \orrightrulename,
   \andleftrulename, \andrightrulename, \impleftrulename,
   \imprightrulename$ are applicable to the conclusion}
    \end{tabular}
}

\newcommand{\newbranchingsteprulename}{\iimp R}
\newcommand{\newbranchingsteprule}{
\AxiomC{$
 \Sigma_l,
 \Theta,
 \opset{\nxt}{\Theta},
 \opset{\limp}{\Gamma},
 \varphi\iimp\psi,
 \varphi
 \seq
 \psi
$}
\LeftLabel{$\newbranchingsteprulename$}
\RightLabel{$\ddagger$}
\UnaryInfC{$
 \Sigma_l,
 \opset{\nxt}{\Theta},
 \opset{\iimp}{\Gamma}
 \seq
 \varphi\iimp\psi,
 \opset{\iimp}{\Delta},
  \opset{\nxt}{\Phi},
 \Sigma_r
$}
\DisplayProof
}

\newcommand{\nxtrightrulename}{\nxt\mathrm{R}}

\newcommand{\newnxtrightrule}{
\AxiomC{$
 \Sigma_l,
 \Theta,
 \opset{\nxt}{\Theta},
 \opset{\limp}{\Gamma},
 \nxt\psi
 \seq
 \psi
$}
\LeftLabel{$\nxtrightrulename$}
\RightLabel{$\ddagger$}
\UnaryInfC{$
 \Sigma_l,
 \opset{\nxt}{\Theta},
 \opset{\iimp}{\Gamma}
 \seq
 \nxt\psi,
 \opset{\iimp}{\Delta},
 \opset{\nxt}{\Phi},
 \Sigma_r
$}
\DisplayProof
}

\begin{document}
\frontmatter          % for the preliminaries
\pagestyle{headings}  % switches on printing of running heads
\mainmatter
\title{Sequent Calculus in the Topos of Trees}

\titlerunning{Running title}  % abbreviated title (for running head)
\author{Ranald Clouston\inst{1} \and Rajeev Gor\'e\inst{2}}
\authorrunning{author list} % abbreviated author list (for running head)
\institute{Department of Computer Science, Aarhus University \\ \email{ranald.clouston@cs.au.dk} \and
  Research School of Computer Science, Australian National University \\ \email{rajeev.gore@anu.edu.au} }

\maketitle              % typeset the title of the contribution

\begin{abstract}
 Nakano's ``later'' modality, inspired by G\"{o}del-L\"{o}b
 provability logic, has been applied in type systems and program logics
 to capture guarded recursion.
 Birkedal et al modelled this modality via the internal logic of the
 topos of trees. We show that the semantics of the propositional
 fragment of this logic
 can be given by linear converse-well-founded
 intuitionistic Kripke frames, so this logic is a marriage of
 the intuitionistic modal logic KM and the intermediate logic LC.
 We therefore call this logic $\lcnxt$.
 We give a sound and cut-free complete sequent calculus for 
 $\lcnxt$ via
 a strategy that decomposes implication into its static and
 irreflexive components. 
 Our calculus provides
 deterministic and terminating backward proof-search, yields 
 decidability of the logic and the coNP-completeness of its validity
 problem.
 Our calculus and decision procedure
 can be restricted to drop linearity and hence capture KM.
\end{abstract}
%

%%%%%%%%%%%%%%%%%%%%%%%%%%%%%%%%%%%%%%%%%%
\section{Introduction}

\emph{Guarded recursion}~\cite{Coquand:Infinite} on an infinite data structure
requires that recursive calls be nested beneath constructors. For
example, a stream of zeros can be defined
with the self-reference guarded by the cons:
\begin{lstlisting}
  zeros = 0 : zeros
\end{lstlisting}
Such equations have \emph{unique} solutions and are \emph{productive}: they
compute arbitrarily large prefixes of the infinite structure in finite time,
a useful property in lazy programming.

Syntactic checks do not always play well with higher-order functions; the insight
of Nakano~\cite{Nakano:Modality} is that guarded recursion can be enforced through
the \emph{type system} via an `approximation modality' inspired by
G\"odel-L\"ob provability logic~\cite{boolos-provability-logic}. We
follow Appel et al~\cite{Appel:Very} and call this modality \emph{later}, and use the
symbol $\nxt$. The meaning of $\nxt\tau$ is roughly `$\tau$ one computation step
later'. Type definitions must have their self-reference guarded by later. For
example streams of integers, which we perhaps expect to be defined as
$Stream \cong \mathbb{Z} \times Stream$, are instead
\[
  Stream \cong \mathbb{Z} \times \nxt Stream
\]
Nakano showed that versions of Curry's fixed-point combinator \textbf{Y}, and Turing's
fixed-point combinator likewise, can be typed by the \emph{strong L\"ob axiom}
(see~\cite{Litak:Constructive})
\begin{equation}\label{eq:sLob}
  (\nxt\tau\limp\tau)\to\tau
\end{equation}
Returning to our example, \textbf{Y} can be applied to the function
\[
  \lambda x.\langle 0,x\rangle:\nxt Stream\to \mathbb{Z}\times \nxt Stream
\]
to define the stream of zeros.

Nakano's modality was popularised by the typing discipline for intermediate
and assembly languages of Appel et al~\cite{Appel:Very}, where for certain
`necessary' types a `L\"ob rule' applies which correlates to the strong L\"ob
axiom~\eqref{eq:sLob}. The modality has since been applied in a wide range of ways;
a non-exhaustive but representative list follows.
As a type constructor, $\nxt$ appears in Rowe's type system for
Featherweight Java~\cite{Rowe:Semantic}, the kind system of the System F extension
FORK~\cite{Pottier:Typed}, and in types for functional reactive
programming~\cite{Krishnaswami:Ultrametric}, with applications to graphical user
interfaces~\cite{Krishnaswami:Semantic}. As a logical
connective, $\nxt$ was married to separation logic in~\cite{Hobor:Oracle},
then to higher-order separation logic in~\cite{Bengtson:Verifying}, and to
\emph{step-indexed} logical relations for reasoning about 
programming languages with LSLR~\cite{Dreyer:Logical}.
Thus Nakano's modality is important in various applications in
computer science.

We have so far been coy on precisely what the logic of later is, beyond positing
that $\nxt$ is a modality obeying the strong L\"ob axiom. Nakano cited G\"odel-L\"ob
provability logic as inspiration, but this is a \emph{classical} modal logic with the
\emph{weak} L\"ob axiom $\square(\square\tau\limp\tau)\limp\square\tau$, whereas
we desire intuitionistic implication and the stronger axiom~\eqref{eq:sLob}. In fact
there does exist a tradition of intuitionistic analogues of 
G\"odel-L\"ob logic~\cite{Litak:Constructive}, 
of which Nakano seemed mainly unaware; we will see
that logic with later can partly be understood through this
tradition. In the computer science literature it has been most common
to leave proof theory and search implicit and fix some
concrete semantics; for example see Appel et al's Kripke semantics
of stores~\cite{Appel:Very}. A more abstract and general model can be given via the
internal logic of the \emph{topos of trees} $\trees$~\cite{Birkedal:First}.
This was shown to generalise several previous models for logic with later, such as
the ultrametric spaces of~\cite{Birkedal:Metric,Krishnaswami:Ultrametric}, and provides
the basis for a rich theory of dependent types. We hence take the internal logic of
$\trees$ as a prominent and useful model of logic with later, in which we can study
proof theory and proof search.

In this paper we look at the propositional-modal core of the internal
logic of $\trees$.  This fragment will be seen to have semantics in
\emph{linear intuitionistic Kripke frames} whose reflexive reduction is
\emph{converse-well-founded}.
Linear intuitionistic frames
are known to be captured by the \emph{intermediate} logic Dummett's
%$\lglc$~\cite{dummett-logic,chagrov-zakharyaschev-modal-logic}; the
$\lglc$~\cite{chagrov-zakharyaschev-modal-logic};
the validity of the $\lglc$ axiom in the topos of trees was first observed
by Litak~\cite{Litak:Typing}. Intuitionistic
frames with converse-well-founded reflexive reduction are captured by the
intuitionistic modal logic $\lgkm$, first called
%$I^\Delta$~\cite{Kuznetsov:Proof,Muravitsky:Logic}. 
$I^\Delta$~\cite{Muravitsky:Logic}. Hence the internal
propositional modal logic of the topos of trees is semantically
exactly their combination, which we call $\lcnxt$
(Litak~\cite[Thm. 50]{Litak:Constructive} has subsequently confirmed this
relationship at the level of Hilbert axioms also).
%In fact our
%results will restrict to $\lgkm$, making a contribution to
%proof theory of that logic also.

Our specific contribution is to give a sound and cut-free complete
sequent calculus for $\lcnxt$, and by restriction for $\lgkm$ also, supporting
terminating backwards proof search and hence yielding the decidability
and finite model property of these logics. Our sequent calculus
also establishes the coNP-completeness of deciding validity in
$\lcnxt$.

To our knowledge sequent calculi for intuitionistic G\"odel-L\"ob logics, let alone $\lgkm$
or $\lcnxt$, have not before been investigated, but such proof systems provide a
solid foundation for proving results such as decidability, complexity, and
interpolation, and given an appropriate link between calculus and
semantics can provide explicit, usually finite, counter-models falsifying given non-theorems.

The main technical novelty of our sequent calculus is that we leverage the fact that the
intutionistic accessibility relation is the reflexive closure of the modal relation, by
decomposing implication into a static (classical)
component and a dynamic `irreflexive implication' $\iimp$ that looks forward along the modal relation. In fact, this irreflexive
implication obviates the need for $\nxt$ entirely, as $\nxt\varphi$ is easily seen to be equivalent to $\top\iimp\varphi$. Semantically
the converse of this applies also, as $\varphi\iimp\psi$ is semantically equivalent to $\nxt(\varphi\limp\psi)$%
\footnote{This in turn is equivalent in $\lcnxt$ (but is not in
  $\lgkm$) to $\nxt\varphi\limp\nxt\psi$~\cite[Sec. 3]{Nakano:Modality}.}%
%see Nakano's discussion~\cite[Sec. 3]{Nakano:Modality}.}%
, but the $\iimp$ connective is a necessary part of our calculus. We maintain $\nxt$ as a first-class connective in deference to
the computer science applications and logic traditions from which we draw, but note
that formulae of the form $\nxt(\varphi\limp\psi)$ are common in the literature - see Nakano's $(\limp E)$ rule~\cite{Nakano:Modality}, and even more directly Birkedal and
M{\o}gelberg's $\circledast$ constructor. We therefore suspect that treating $\iimp$
as a first-class connective could be a conceptually fruitful side-benefit of our work.

% DELETE THIS PARAGRAPH FOR ARXIV VERSION!
%Note that for space reasons some proofs appear only in the extended
%version of this paper~\cite{ARXIVVERSION}.

%%%%%%%%%%%%%%%%%%%%%%%%%%%%%%%%%%%%%%%%%%
\section{From the Topos of Trees to Kripke Frames}\label{sec:trees}

In this section we outline the \emph{topos of trees} model and its internal logic, and
show that this logic can be described semantically by conditions on intuitionistic Kripke
frames.
Therefore after this section we discard category theory and proceed with reference to
Kripke frames alone.

The topos of trees, written $\trees$, is the category of presheaves on the first infinite
ordinal $\omega$ (with objects $1,2,\ldots$, rather than starting at $0$, in keeping with
the relevant literature).
Concretely an \emph{object} $A$ is a pair of a family of sets $A_i$
indexed by the positive integers, and a family of \emph{restriction functions}
$\res{A}{i}:A_{i+1}\to A_i$ indexed similarly. An \emph{arrow} $f:A\to B$ is a family
of functions $f_i:A_i\to B_i$ indexed similarly, subject to \emph{naturality}, i.e. all
squares below commute:
\[\xymatrix{
  A_1 \ar[d]_{f_1} & A_2  \ar[l]_{a_1} \ar[d]_{f_2} & A_3  \ar[l]_{a_2}  \ar[d]_{f_3} & \cdots & A_j  \ar[d]_{f_j} & A_{j+1}  \ar[l]_{a_j} \ar[d]_{f_{j+1}} \\
  B_1 & B_2 \ar[l]^{b_1} & B_3 \ar[l]^{b_2} & \cdots & B_j & B_{j+1} \ar[l]^{b_j}
}\]
Two $\trees$-objects are of particular interest: the \emph{terminal object} $1$ has 
singletons as component sets and identities as restriction functions;
the \emph{subobject classifier} $\Omega$ has $\Omega_j=\{0,\ldots,j\}$ and
$\omega_j(k)=min(j,k)$. We regard the positive integers as \emph{worlds} and
functions $x:1\to\Omega$ as \emph{truth values} over these worlds, by
considering $x$ true at $j$ iff $x_j=j$. Such an $x$ is constrained by naturality to have
one of three forms: $x_j=j$ for all $j$ (\emph{true everywhere}); $x_j=0$ for all $j$ (\emph{true nowhere}); or given any positive integer $k$, $x_j$ is $k$ for all $j\geq k$,
and is $j$ for all $j\leq k$ (\emph{becomes true at world $k$, remains true at all lesser
worlds}). As such the truth values can be identified with the set $\mathbb{N}\cup
\{\infty\}$, where $\infty$ captures `true everywhere'.

Formulae of the internal logic of $\trees$ are defined as
\begin{eqnarray*}
  \varphi & ::= &
      p \mid
     \top \mid
     \bot \mid
      \varphi\land\varphi \mid
     \varphi\lor\varphi \mid
     \varphi\limp\varphi \mid
     \varphi\iimp\varphi \mid
     \nxt\varphi
\end{eqnarray*}
where $p\in\atms$ is an atomic formula.
Negation may be defined as usual as $\varphi\limp\bot$.
The connective $\iimp$, read as \emph{irreflexive implication}, is not in Birkedal et
al~\cite{Birkedal:First} but is critical to the sequent calculus of this paper; readers may
view $\iimp$ as a second-class connective generated and then disposed of
by our proof system, or as a novel first-class connective, as they prefer.

Given a map $\eta$ from propositional variables $p\in\atms$ to arrows $\eta(p):1\to\Omega$, and a positive integer $j$, the
Kripke-Joyal forcing semantics for $\trees$ are defined by
\[\begin{array}{ll}
  \eta,j\forces p &\mbox{iff $\eta(p)_j=j$} \\
  \eta,j\forces \top & \mbox{always} \\
  \eta,j\forces \bot & \mbox{never} \\
  \eta,j\forces \varphi\land\psi \quad &
    \mbox{iff }\eta,j\forces\varphi\mbox{ and }\eta,j\forces\psi \\
  \eta,j\forces \varphi\lor\psi \quad &
    \mbox{iff }\eta,j\forces\varphi\mbox{ or }\eta,j\forces\psi \\
  \eta,j\forces \varphi\limp\psi \quad &
    \mbox{iff $\forall k\leq j$. }
    \eta,k\forces\varphi\mbox{ implies }\eta,k\forces\psi \\
  \eta,j\forces \varphi\iimp\psi \quad &
    \mbox{iff $\forall k< j$. }
    \eta,k\forces\varphi\mbox{ implies }\eta,k\forces\psi \\
  \eta,j\forces \later\varphi &
    \mbox{iff $\forall k< j$. }
    \eta,k\forces\varphi
\end{array}\]

A formula $\varphi$ is \emph{valid} if $\eta,j\forces\varphi$ for all $\eta,j$.
Note that $\varphi\iimp\psi$ is equivalent to $\later(\varphi\limp\psi)$, and $\later\varphi$ is equivalent to $\top\iimp\varphi$.
While implication $\limp$ can be seen as a conjunction of static and irreflexive components:
\begin{equation}\label{eq:limp_via_imp}
    j\forces\varphi\limp\psi
      \mbox{ \ iff \ ($j\forces\varphi$ implies $j\forces\psi$) and $j\forces\varphi\iimp\psi$}
\end{equation}
it is not definable from the other connectives, because we have no static (that is,
classical) implication. However our sequent calculus will effectively capture
\eqref{eq:limp_via_imp}.

We now turn to Kripke frame semantics. Kripke semantics for intuitionistic modal
logics are usually defined via bi-relational frames $\langle W, R_\limp,
R_\wbx \rangle$, where $R_\limp$ and $R_\wbx$ are binary relations on $W$, with
certain interaction conditions ensuring that modal formulae persist along the
intuitionistic relation~\cite{wolter-zakharyaschev-iml}. However for $\lgkm$ and $\lcnxt$ the intuitionistic relation is definable in terms of the box relation, and so only
the latter relation need be explicitly given to define a frame:

\begin{definition}\label{def:frame_conds}
A \emph{frame} is a pair $\langle W,\rel\rangle$ where $W$ is a non-empty set and
$\rel$ a binary relation on $W$. A \emph{$\lgkm$-frame} has $\rel$ transitive and
\emph{converse-well-founded}, i.e.
there is no infinite sequence $x_1\rel x_2 \rel x_3 \rel \cdots$.
A \emph{$\lcnxt$-frame} is a $\lgkm$-frame with $\rel$ also \emph{connected}, i.e.
$\forall x,y\in W.~x=y$ or $\rel(x,y)$ or $\rel(y,x)$.
\end{definition}

Converse-well-foundedness implies irreflexivity. Also,
$\lgkm$- and $\lcnxt$-frames may be infinite because non-well-founded chains
$\cdots \rel w_3 \rel w_2 \rel w_1$ are permitted.

Given a binary relation $\rel$, let $\releq$ be its \emph{reflexive closure}. If
$\langle W,\rel \rangle$ is a $\lgkm$-frame then $\langle W,\releq \rangle$ is reflexive
and transitive so provides frame semantics for intuitionistic logic. In fact frames arising
in this way in general satisfy only the theorems
of intuitionistic logic, so $\lgkm$ is conservative over intuitionistic logic. In other words,
the usual propositional connectives are too coarse to detect the converse
well-foundedness of a frame; for that we need $\nxt$ and the strong
L\"ob axiom~\eqref{eq:sLob}. Similarly the reflexive closure of a $\lcnxt$-frame is a
\emph{linear} relation and so gives semantics for the logic $\lglc$, over which $\lcnxt$
is conservative.

%The frame conditions of Def.~\ref{def:frame_conds} are known in \emph{classical} %modal logic as those for the logics
%$\lgsf$, $\lgsfpt$, $\lggl$ and $\lgglpt$ respectively, while the reflexive closure of %$\lgkm$- (resp. $\lcnxt$-) frames define the
%classical modal logic $\lggrz$ (resp. $\lggrzpt$). We return to the relationships with %these logics in Sec.~\ref{sec:embed}.

A \emph{model} 
$\langle W, \rel, \vartheta \rangle$
%is a triple where 
consists of a frame
$\langle W, \rel \rangle$ 
%is a frame 
and a valuation 
$\vartheta: \atms\mapsto 2^{W}$ 
%is a valuation 
%which maps 
%mapping each atomic formula to the 
%set of
%worlds where it is true, and 
obeying
\textbf{persistence}:
\[
\mathrm{if~} 
w \in \vartheta(p) 
\mathrm{~and~}
w \rel x
\mathrm{~then~}
x \in \vartheta(p) 
\]
We hence define $\lgkm$- and $\lcnxt$-models by the relevant frame conditions.

We can now define when a $\lgkm$- or $\lcnxt$-model $M=\langle W, \rel, \vartheta \rangle$ makes a formula true at a world $w\in W$, with obvious cases $\top,\bot,\land,\lor$ omitted:
$$\begin{array}{ll}
  M,w\forces p &\mbox{iff $w\in\vartheta(p)$} \\
%  M,w\forces \top & \mbox{always} \\
%  M,w\forces \bot & \mbox{never} \\
%  M,w\forces \varphi\land\psi \quad &
%    \mbox{iff }M,w\forces\varphi\mbox{ and }M,w\forces\psi \\
%  M,w\forces \varphi\lor\psi \quad &
%    \mbox{iff }M,w\forces\varphi\mbox{ or }M,w\forces\psi \\
  M,w\forces \varphi\limp\psi \quad &
    \mbox{iff } \forall x.\,w\releq x\mbox{ and }
    M,x\forces\varphi\mbox{ implies }M,x\forces\psi \\
  M,w\forces \varphi\iimp\psi \quad &
    \mbox{iff } \forall x.\,w\rel x\mbox{ and }
    M,x\forces\varphi\mbox{ implies }M,x\forces\psi \\
  M,w\forces \later\varphi &
    \mbox{iff } \forall x.\,w\rel x\mbox{ implies }
    M,x\forces\varphi
\end{array}$$
Thus $\later$ is the usual modal box. 
%A modal diamond is of no interest to us in this paper.
As usual for intuitionistic logic, we have a monotonicity lemma, provable
by induction on the formation of $\varphi$:

\begin{lemma}[Monotonicity]
\label{lemma:monotonicity}
If
$M, w \forces \varphi$
and
$w \rel v$ then
$M, v \forces \varphi$.
\end{lemma}

Fixing a class of models ($\lgkm$- or $\lcnxt$-),
a formula $\varphi$ is 
%\textbf{satisfiable}
%iff there exists a model $M$ containing a world $w$ such that
%$M, w \forces \varphi$.
%A formula $\varphi$ is 
%\textbf{refutable}
%iff there exists a model $M$ containing a world $w$ such that
%$M, w \not\forces \varphi$.
%A formula $\varphi$ is 
\emph{valid}
if for every world $w$ in every model $M$ we have
$M, w \forces \varphi$.
%For a finite set $T$ of formulae,
%and a model
%$\langle W, \rel, \vartheta \rangle$,
%we write 
%$M \forces T$ if for all $w \in W$, and for all
%$\psi \in T$, we have
%$M, w \forces \psi$.
%A formula
%$\varphi$ is \textbf{global logical consequence}\footnote{We don't
%  actually consider this, but only consider validity. Probably we
%  should just remove this part.}
%of a finite set $T$ of formulae if:
%\[
%\forall M. \mbox{ if } M \forces T \mbox{ then } M \forces \varphi
%\]
It is easy to observe that the two semantics presented above coincide, given the right choice of frame conditions:
\begin{theorem}
Formula $\varphi$ is valid in the internal logic of $\trees$ iff it is $\lcnxt$-valid.
\end{theorem}
% \begin{theorem}
% A formula $\varphi$ is valid in the internal logic of $\trees$ iff it is valid with respect to $\lcnxt$-models.
% \end{theorem}

%%%%%%%%%%%%%%%%%%%%%%%%%%%%%%%%%%%%%%%%%%
\section{The Sequent Calculus $\seqlcnxt$ for $\lcnxt$}

A \emph{sequent} is an expression of the form
$\Gamma \seq \Delta $
where 
$\Gamma$
and
$\Delta$
are finite, possibly empty, sets of formulae with
$\Gamma$ the \emph{antecedent} and 
$\Delta$ the \emph{succedent}.
We write
$\Gamma, \varphi$
for
$\Gamma \cup \{\varphi\}$.
Our sequents are
``multiple-conclusioned'' since the succedent $\Delta$ is a
finite set rather than a single formula as in ``single-conclusioned''
sequents. 

A sequent \emph{derivation} is a finite tree of sequents where each
internal node is obtained from its parents by instantiating a
rule. The root of a derivation is the
\emph{end-sequent}. A sequent derivation is a \emph{proof} if all
the leaves are zero-premise rules.  A rule may require extra
side-conditions for its (backward) application.
%it can be applied (upwards).

The sequent calculus $\seqlcnxt$ is shown in
Fig.~\ref{fig:seq-rules}, where $\Gamma$, $\Delta$, $\Phi$, $\Theta$, and $\Sigma$,
with superscripts and/or subscripts, are finite, possibly empty, sets of formulae.
% Also, if $\Theta = \{\theta_1 ,
% \cdots , \theta_n\}$ then 
% $\nxt\Theta = \{\nxt\theta_1 , \cdots ,
% \nxt\theta_n\}$.

\begin{figure}[h!]
  \centering
  \begin{tabular}[c]{l@{\extracolsep{0.5cm}}l}
    \multicolumn{2}{l}{
      $\quad\toprule \qquad\qquad \idrule \qquad\qquad \botrule$}
\\[15px]
    \orleftrule &  \orrightrule
\\[15px]
    \andleftrule  & \andrightrule
\\[15px]
    \impleftrule & \imprightrule
\\[15px]
    \multicolumn{2}{l}{
    \newsteprulenewversion
    }
  \end{tabular}
  \caption{Rules for sequent calculus $\seqlcnxt$}
  \label{fig:seq-rules}
\end{figure}

Rules $\toprightrulename$, $\botleftrulename$, $\idrulename$,
$\orleftrulename$, $\orrightrulename$, $\andleftrulename$,
$\andrightrulename$ are standard for a multiple-conclusioned calculus
for $\lgint$~\cite{troelstra-schwichtenberg-basic}.  Rules
$\impleftrulename$ and $\imprightrulename$ can be seen as branching on
a conjunction of static and an irreflexive
implication: see equation~\eqref{eq:limp_via_imp}.
The occurrence of 
$\varphi\iimp\psi$
in the right premise of $\impleftrulename$ is redundant, since 
$\psi$ implies $\varphi\iimp\psi$, but its presence makes our
termination argument simpler.

The rule $\steprulename$ resembles Sonobe's multi-premise rule for
$\imprightrulename$ in $\lglc$~\cite{Sonobe:Gentzen,Corsi:Semantic},
but its interplay of static and dynamic connectives allows us to
capture the converse-well-foundedness of our frames.  The reader may
like to skip forward to compare it to the rules for $\lgkm$ in
Fig.~\ref{fig:step-rules-logic-KH}, which are simpler because they do
not have to deal with linearity. Condition $\condzero$ is essential
for soundness; $\condone$ and $\condtwo$ are not, but ensure that the
$\newsteprulename$ rule is applicable only if no other rules are
applicable (upwards), which is necessary for semantic invertibility
(Lem.~\ref{lemma-invertibility-steprule}).  Note that the formulae in
$\opset{\nxt}{\Theta}$ appear intact in the antecedent of every
premise. This is not essential as $\Theta$ implies
$\opset{\nxt}{\Theta}$, but will simplify our proof of
completeness. In constrast the formulae in $\opset{\nxt}{\Phi}$ do not
appear in the succedent of any premise. Also, the formulae in
$\Sigma_r$ do not appear in the succedent of any premise. So
$\newsteprulename$ contains two aspects of weakening, but $\condtwo$
ensures this is not done prematurely.
% Thus
% there are two aspects of weakening built into the $\newsteprulename$
% rule, but $\condtwo$ ensures that this is not done prematurely.

Figs.~\ref{fig:sLob} and~\ref{fig:LC} give example proofs,
using the following derived rule:

\begin{lemma}\label{Lem:MP}
  The Modus Ponens rules $\mathrm{\small \mprulename}$ 
  is derivable in $\seqlcnxt$ as follows:
\end{lemma}
\begin{proof}
  \[
  \AxiomC{}
  \RightLabel{\idrulename}
 \UnaryInfC{$\Gamma,\varphi,\varphi\iimp\psi \seq \varphi,\psi$}
 \AxiomC{}
  \RightLabel{\idrulename}
 \UnaryInfC{$\Gamma,\varphi,\varphi\iimp\psi,\psi \seq\psi$}
  \RightLabel{$\impleftrulename$}
 \BinaryInfC{$\Gamma,\varphi,\varphi\limp\psi \seq \psi$}
 \DisplayProof
  \]
\end{proof}

% \begin{proof}
%   \[
%   \AxiomC{}
%   \RightLabel{\idrulename}
%  \UnaryInfC{$\Gamma,\varphi,\varphi\iimp\psi \seq \varphi,\psi$}
%  \AxiomC{}
%   \RightLabel{\idrulename}
%  \UnaryInfC{$\Gamma,\varphi,\varphi\iimp\psi,\psi \seq\psi$}
%   \RightLabel{$\impleftrulename$}
%  \BinaryInfC{$\Gamma,\varphi,\varphi\limp\psi \seq \psi$}
%  \DisplayProof
%   \]
% \end{proof}

\begin{figure}[t]
  \centering
   \AxiomC{}\RightLabel{\mprulename}
   \UnaryInfC{$(\nxt p \limp p) \limp p, \nxt p \limp p, \nxt p  \seq p$}
   \RightLabel{$\steprulename$}
   \UnaryInfC{$(\nxt p \limp p) \iimp p, \nxt p \iimp p \seq \nxt p , p$}
   \AxiomC{}\RightLabel{\idrulename}
   \UnaryInfC{$(\nxt p \limp p) \iimp p, \nxt p \iimp p, p \seq p$}
   \RightLabel{$\impleftrulename$}
   \BinaryInfC{$(\nxt p \limp p) \iimp p, \nxt p \limp p \seq p$}
   \RightLabel{$\steprulename$}
   \UnaryInfC{$\seq (\nxt p \limp p) \iimp p$}
   \DisplayProof
   \\[15px]
   \AxiomC{}\RightLabel{\mprulename}
   \UnaryInfC{$\nxt p \limp p, \nxt p  \seq p$}
   \RightLabel{$\steprulename$}
   \UnaryInfC{$\nxt p \iimp p \seq \nxt p , p$}
   \AxiomC{}\RightLabel{\idrulename}
   \UnaryInfC{$\nxt p \iimp p, p \seq p$}
   \RightLabel{$\impleftrulename$}
   \BinaryInfC{$\nxt p \limp p \seq p$}
   \AxiomC{$\seq (\nxt p \limp p) \iimp p$}
   \RightLabel{$\imprightrulename$}
   \BinaryInfC{$\seq (\nxt p \limp p) \limp p$}
   \DisplayProof   
  \caption{$\seqlcnxt$ proof of the strong L\"ob axiom}
  \label{fig:sLob}
\end{figure}

\begin{figure}[t]
  \centering
  \def\ScoreOverhang{2.7pt}
  \AxiomC{}\RightLabel{\idrulename}
  \UnaryInfC{$p \iimp q, p , q \seq p, q$}
  \AxiomC{}\RightLabel{\idrulename}
  \UnaryInfC{$p \limp q, p, q \iimp p, q \seq p$}
  \RightLabel{$\steprulename$}
  \UnaryInfC{$p \iimp q, p \seq q , q \iimp p$}
  \RightLabel{$\imprightrulename$}
  \BinaryInfC{$p \iimp q, p \seq q , q \limp p$}
  \AxiomC{$\mbox{Symmetric to left}$}
  \UnaryInfC{$q \iimp p, q \seq  p, p \limp q$}
  \RightLabel{$\steprulename$}
  \BinaryInfC{$\seq p \iimp q , q \iimp p$}
  \DisplayProof
 \\[15px]
  \AxiomC{}\RightLabel{\idrulename}
  \UnaryInfC{$p, q \seq q , p$}
  \AxiomC{}\RightLabel{\idrulename}
  \UnaryInfC{$p, q \iimp p, q \seq p$}
  \RightLabel{$\steprulename$}
  \UnaryInfC{$p \seq q , q \iimp p$}
  \RightLabel{$\imprightrulename$}
  \BinaryInfC{$p \seq q , q \limp p$}
  \AxiomC{}\RightLabel{\idrulename}
  \UnaryInfC{$q, p \iimp q, p \seq q$}
  \RightLabel{$\steprulename$}
  \UnaryInfC{$q \seq p, p \iimp q$}
  \AxiomC{$\seq p \iimp q , q \iimp p$}
  \RightLabel{$\imprightrulename$}
  \BinaryInfC{$\seq p \iimp q , q \limp p$}
  \RightLabel{$\imprightrulename$}
  \BinaryInfC{$\seq p \limp q , q \limp p$}
  \RightLabel{$\orrightrulename$}
  \UnaryInfC{$\seq p \limp q \lor q \limp p$}
  \DisplayProof
  \caption{$\seqlcnxt$ proof of the $\lglc$ axiom}
  \label{fig:LC}
\end{figure}

%\section{Soundness of $\seqlcnxt$}
\subsection{Soundness of $\seqlcnxt$}

Given a world $w$ in some model $M$, and finite sets $\Gamma$ and
$\Delta$ of formulae, we write $w \forces \Gamma$ 
if every formula in $\Gamma$ is true at $w$ in model $M$ 
and write 
$w \not\forces \Delta$
if every formula in $\Delta$ is not true at $w$ in model $M$.

A sequent  $\Gamma \seq \Delta $ 
is \textbf{refutable} if there exists a model $M$ and a world $w$ in that
model such that
$w \forces \Gamma$
and $w \not\forces\Delta$.
A sequent  
is \textbf{valid} if it is not refutable.
A rule is \textbf{sound} 
if some premise is refutable whenever the conclusion is refutable.
A rule is \textbf{semantically invertible} 
if the conclusion is refutable whenever some premise is
refutable. 
Given a model $M$ and a formula $\varphi$, a world $w$ is a 
\textbf{refuter} 
for $\varphi$ if $M, w \not\forces \varphi$. It is a 
\textbf{last refuter} for $\varphi$ if in addition 
$M, w \forces \nxt\varphi$.
An \textbf{eventuality} is a formula of the form 
$\varphi\iimp\psi$  or
$\nxt\varphi$ in the succedent of the conclusion of an
application  of the rule 
$\steprulename$.

\begin{lemma}~\label{last-refuter}
  In every model, every formula $\varphi$ with a refuter has a last refuter.
  % In every model, every formula $\varphi$ has a last refuter
  % if it has a refuter.
\end{lemma}
\begin{proof}
Suppose $\varphi$ has refuter $w$ in model $M$, i.e.
$M, w \not\forces \varphi$.
If all $\rel$-successors $v$ of $w$ have $v\forces\varphi$ then
$w \forces \nxt\varphi$, and so $w$ is the last refuter we seek.
Else pick any successor $v$ such that $M, v \not\forces \varphi$
and repeat the argument replacing $w$ with $v$.
By converse well-foundedness this can only be done finitely often before
reaching a world with no $\rel$-successors, which vacuously satisfies
$\nxt\varphi$.
\end{proof}

%\begin{lemma}\label{Lem:Soundness_Lemma}
%Every rule of $\seqlcnxt$ is sound.
%\\
%OR
%\\
% For every rule of $\seqlcnxt$ , if the conclusion is refutable then so is some
% premise.
%\end{lemma}

\begin{theorem}[Soundness]\label{thm:soundness}
%  For all formulae $\varphi$,
  If
% the sequent 
  $\seq \varphi$
  is $\seqlcnxt$-derivable then 
%the formula 
$\varphi$ is $\lcnxt$-valid.
\end{theorem}
\begin{proof}
We consider only the non-standard rules.
  \begin{description}
  \item[$\imprightrulename$:] 
  Suppose the conclusion
  $\Gamma \seq \varphi \limp \psi, \Delta $ is refutable
  at $w$ in model $M$.
  Thus some $\releq$-successor $v$ of $w$ refutes
  $\varphi \limp \psi$ via
  $M, v \forces \varphi$ and
  $M, v \not\forces \psi$.
  If
  $v=w$ then
  $w$ refutes the left premise
  $\Gamma, \varphi \seq \psi, \Delta$.
  Else $w \rel v$ and 
  $M, w \not\forces \varphi\iimp\psi$,
  so $w$ refutes the right premise
  $\Gamma \seq \varphi \iimp \psi, \Delta$.
  \item[$\impleftrulename$:] 
  Suppose the conclusion 
  $\Gamma, \varphi\limp\psi \seq \Delta$
  is refuted at $w$. 
  Hence
  $w \forces \Gamma$ and
  $w \forces \varphi\limp\psi$ and
  $w \not\forces \Delta$.
  Thus
  $w \forces \varphi\iimp\psi$.
  If 
  $w \forces \psi$ 
  then 
  $w$ refutes the right premise
  $\Gamma, \varphi\iimp\psi, \psi \seq \Delta$.
  Else
  $w \not\forces \psi$ 
  and so we must have
  $w \not\forces \varphi$
  since we already know that
  $w \forces \varphi\limp\psi$.
  Thus $w$ refutes the left premise
  $\Gamma, \varphi\iimp\psi \seq \varphi, \Delta$.
  
  \item[\rm $\steprulename$:] Assume that
     $\Sigma_l, 
      \opset{\nxt}{\Theta},
      \opset{\iimp}{\Gamma}
      \seq
      \opset{\iimp}{\Delta},
      \opset{\nxt}{\Phi},
      \Sigma_r$
     is refutable.
     That is, there is some model $M$ and some world $w$ such
     that 
    $M, w \forces \Sigma_l$ and
    $M, w \forces \opset{\nxt}{\Theta}$ and
    $M, w \forces \opset{\iimp}{\Gamma}$ 
    but 
    $M, w \not\forces \opset{\iimp}{\Delta}$ 
    and
    $M, w \not\forces \opset{\nxt}{\Phi}$ 
    and
    $M, w \not\forces \Sigma_r$.

    Thus each $\nxt\phi_i \in \opset{\nxt}{\Phi}$
    and
    each $\varphi_i \iimp \psi_i \in \opset{\iimp}{\Delta}$
    has a last
    refuter, which may be $w$ itself.
    But then, 
    each $\phi_i \in \Phi$ and
    each $\varphi_i \limp \psi_i \in \opset{\limp}{\Delta}$,
    has a last
    refuter which 
    is a strict successor of $w$.
    From this set of strict successors of $w$,
    choose the refuter $v$  that is
    closest to $w$ in the linear order.

    Since $w \rel v$, we must have
    $M, v \forces \Sigma_l$ and
    $M, v \forces \Theta$ and
    $M, v \forces \opset{\nxt}{\Theta}$ and
    $M, v \forces \opset{\limp}{\Gamma}$, giving that
    $M, v \forces \Sigma_l, \Theta, \opset{\nxt}{\Theta}, \opset{\limp}{\Gamma}$.
%    Moreover, $v$ must be a last refuter for some
%    $\varphi_i \iimp \psi_i \in \opset{\iimp}{\Delta}$
%    or some
%   $\nxt\phi_i \in \opset{\nxt}{\Phi}$
%   since
%    $\opset{\iimp}{\Delta} \cup \opset{\nxt}{\Phi} \neq \emptyset$.

    If $v$ is the last refuter for
    some 
    $\varphi_i \limp \psi_i \in \opset{\limp}{\Delta}$,
    we must have 
    % $M, v \forces \varphi$ and
    % $M, v \not\forces \psi$ and
    $M, v \forces \varphi_i$ and
    $M, v \not\forces \psi_i$ and
    $M, v \forces \varphi_i\iimp\psi_i$. 
    We must also have 
    $M, v \not\forces \Phi$ since the last refuter for each
    $\phi_i \in \Phi$ cannot strictly precede $v$,
    by our choice of $v$.
    For the same reason, we must have
    $M, v \not\forces \varphi_j\limp\psi_j$ for every
    $1 \leq j \neq i \leq k$, giving
    $M, v \not\forces\opsetmi{\limp}{\Delta}{i}$.
    Thus $v$ refutes the 
    $i^\mathrm{th}$ premise 
    $\mathrm{Prem}_i = 
     \Sigma_l, \opset{\limp}{\Gamma}, \Theta, \opset{\nxt}{\Theta},
     \varphi_i\iimp\psi_i , \varphi_i \seq \psi_i, 
     \opsetmi{\limp}{\Delta}{i}, \Phi
    $.

    If $v$ is the last refuter for
    some
    $\phi_i \in \Phi$,
    we must have both
    $M, v \not\forces \phi_i$ and
    $M, v \forces \nxt\phi_i$.
    Since
    $v$ is the closest last refuter to $w$ in the linear order,
    the last refuters for the other formulae in $\Phi$ cannot strictly
    precede $v$.
    Hence for each $1 \leq j \neq i \leq n$,
    we must have
    $M, v \not\forces \phi_j$
    for each 
    $\phi_j \in \Phi$, hence
    $M, v \not\forces \Phi$.
    Moreover, for the same reason,
    we must have
    $M, v \not\forces \varphi_j \limp \psi_j$,
    where $1 \leq j \leq k$,
    for each 
    $\varphi_j\iimp\psi_j \in \opset{\iimp}{\Delta}$, hence
    $M, v \not\forces \opset{\limp}{\Delta}$.
    That is, $v$ refutes the
    $(k+i)$-th premise
    $\mathrm{Prem}_{k+i} = \Sigma_l, \Theta, 
     \opset{\nxt}{\Theta}, \opset{\limp}{\Gamma},
     \nxt \phi_i \seq \opset{\limp}{\Delta},
     \Phi$.
   \end{description}
\end{proof}

%%%%%%%%%%%%%%%%%%%%%%%%%%%%%%%%%%%%%%%%%%

%\section{Terminating backward proof search}\label{Sec:Proof_Search}
\subsection{Terminating backward proof search}\label{Sec:Proof_Search}

In this section we describe how to systematically find derivations using backward proof
search. To this end, we divide the rules into three sets as follows:
\begin{description}
\item[\rm \emph{Termination Rules:}] the rules 
$\idrulename,
\botleftrulename,
\toprightrulename$ 
\item[\rm \emph{Static Rules:}] the rules
      $\impleftrulename,
      \imprightrulename,
      \orleftrulename,
      \orrightrulename,
      \andleftrulename,
      \andrightrulename$
\item[\rm \emph{Transitional Rule:}] $\steprulename$.
\end{description}

The proof search strategy below starts at the leaf (end-sequent)
  $\Gamma_0 \seq \Delta_0$: \vspace{2mm} 
\\
\textbf{while} some rule is applicable to a leaf sequent \textbf{do} \\
  \hphantom{11}\textbf{stop:} apply any applicable termination rule to that leaf \\
  \hphantom{11}\textbf{saturate:}  else apply any applicable static rule to that leaf \\
%  \hphantom{11}\textbf{step:} else apply the transitional rule to
%  that leaf \vspace{2mm}
  \hphantom{11}\textbf{transition:} else apply the transitional rule to that leaf \vspace{2mm}

The phase where only static
rules are applied is called the \textbf{saturation} phase. The
only non-determinism in our procedure is the choice of static rule
when many static rules are applicable, but as we shall see later, any
choice suffices. Note that conditions $\condone$ and $\condtwo$ actually
force $\newsteprulename$ to have lowest priority.

 Let 
 $sf(\varphi)$ 
 be the set of subformulae of $\varphi$, including 
 $\varphi$ itself and let $m$ be the length of $\varphi$.
 Let 
 $cl(\varphi) = sf(\varphi) 
     \cup 
   \{\psi_1\iimp\psi_2~\mid~
     \psi_1\limp\psi_2 \in sf(\varphi)\}$.

\begin{proposition}\label{prop:saturation-terminates}
  The (backward) saturation phase terminates for any sequent.
\end{proposition}
\begin{proof}
Each rule either: removes a connective; or removes a formula completely; or replaces a
formula $\varphi\limp\psi$ with $\varphi\iimp\psi$ to which no static rule can be
applied.
\end{proof}

Given our strategy (and condition $\condone$), 
we know  that the conclusion of the
$\steprulename$
rule will never  be an instance of 
$\idrulename$, 
hence 
$\varphi\iimp\psi$  or
$\nxt\varphi$
is only an eventuality when an occurrence of it does not already
appear in the antecedent of the conclusion of the 
$\steprulename$ rule
in question. 
% Moreover, we also know that the conclusion of each
% backward application of the $\steprulename$ rule must be saturated.
% So we can add the following further conditions to the
% $\steprulename$ rule, which automatically guarantees that its
% conclusion is not an instance of 
% a termination rule:
% % the $\idrulename$ rule:
% \begin{description}
% \item[\textbf{C1:}]
%     $\bot \not\in\Sigma_l$ and $\top\not\in\Sigma_r$ and 
%     $(\Sigma_l \cup \opset\nxt\Theta \cup \Gamma^\iimp ) \cap 
%        (\opset{\iimp}{\Delta} 
%         \cup 
%         \opset{\nxt}{\Phi} 
%         \cup 
%         \Sigma_r)  = \emptyset$
% \end{description}
%
%Note that 
%The rule $\steprulename$
%inserts $\nxt$-formulae occurrences into the antecedent of some premises, but
%these are not eventualities by our definition since eventualities must
%appear in a succedent. 

\begin{proposition}\label{prop:nonew}
For all rules, the formulae in the premise succedents
are subformulae of formulae in the
conclusion, or are $\limp$-formulae created from $\iimp$-formulae
in the conclusion succedent: we never create new eventualities upwards.
\end{proposition}

\begin{proposition}\label{prop:reduces}
  Any application of the rule $\steprulename$ has strictly fewer eventualities in each
  premise, than in its conclusion.
\end{proposition}
\begin{proof}
  For each premise, an
  eventuality $\nxt\varphi$ crosses from the succedent of the
  conclusion to the antecedent of that premise and appears in all
  higher antecedents, or an eventuality $\varphi\iimp\psi$ from the
  succedent of the conclusion turns into $\varphi\limp\psi$ in the
  antecedent of the premise and this $\varphi\limp\psi$ 
%itself 
  turns back into $\varphi\iimp\psi$ via saturation, meaning that the
  eventuality ($\nxt\varphi$ or $\varphi\iimp\psi$) cannot reappear in
  the succedent of some higher {\em saturated} sequent without creating an instance of $\idrulename$.
\end{proof}

\begin{theorem}
 Backward proof search terminates.
\end{theorem}

\begin{proof}
  By Prop.~\ref{prop:saturation-terminates} each saturation
  phase terminates, so the only way a branch can be infinite is
  via an infinite number of applications of the $\steprulename$
  rule. But by Prop.~\ref{prop:reduces} each such application
  reduces the number of eventualities of the branch, and by
  Prop.~\ref{prop:nonew}, no rule creates new eventualities.
  % since the eventuality that crosses from the succedent of the
  % conclusion to the antecedent of the premise is either a
  % $\nxt$-formula or a $\iimp$-formula, and such formulae are
  % persistent in the antecedents of the conclusion of every instance of
  % $\steprulename$ that appears higher in the branch. That is, once
  % such a formula enters the antecedent of a premise of
  % $\steprulename$, it never disappears (upwards): $\iimp$-formulae may
  % turn into $\limp$-formulae, but they then must turn back into
  % $\limp$-formulae via $\impleftrulename$ during saturation. So if the
  % eventuality ever re-appears in a higher succedent, that sequent is
  % an instance of $\idrulename$ (to which the $\steprulename$ rule is
  % no longer applicable due to \textbf{C1}).
  Thus we must eventually reach a saturated sequent to which no rule
  is applicable, or reach an instance of a termination rule. Either
  way, proof search terminates.
\end{proof}

\begin{proposition}
  Given an end-sequent $\Gamma_0 \seq \Delta_0$, the maximum number of
  different eventualities is the sum of the lengths of the formula in
  $\Gamma_0 \cup \Delta_0$.
\end{proposition}
\begin{proof}
  Each eventuality $\nxt\varphi$ is a
  subformula of the end-sequent, and each eventuality
  $\varphi\iimp\psi$ is created from a subformula $\varphi\limp\psi$
  which is also a subformula of the end-sequent or is a subformula of
  the end-sequent.
\end{proof}

\begin{corollary}\label{cor:onlylsteprules}
  Any branch of our proof-search procedure for end-sequent 
  $\Gamma_0\seq\Delta_0$ contains at most $l$ applications of the 
  $\steprulename$ rule, where $l$ is the sum of the lengths of the
  formulae in $\Gamma_0 \cup \Delta_0$. 
\end{corollary}

% \begin{proof}
%   The maximum number of different eventualities that can ever be
%   brought to the surface by backward rule applications is bounded by
%   $l$.
% \end{proof}

\subsection{Cut-free Completeness Without Backtracking}

%The usual semantic route to completeness is via a Hintikka-structure
%built from downward-saturated sequents. But
The rules of our sequent calculus,
when used according to conditions $\condzero$, $\condone$, and $\condtwo$,
can be shown to preserve validity upwards
%by showing that they preserve refutability downwards
as follows.

\begin{lemma}[Semantic Invertibility]\label{lemma-invertibility-static-rules}
  All static rules are semantically invertible: if some premise is refutable
  then so is the conclusion.
\end{lemma}
\begin{proof}
  Again, we consider only the non-standard rules.
  \begin{description}
  \item[$\imprightrulename$:] Suppose the right premise
   $\Gamma \seq \varphi\iimp\psi , \Delta $ is refuted at $w$.
   Then so is the conclusion
   $\Gamma \seq \varphi\limp\psi , \Delta $ 
   since $\varphi\limp\psi$ implies 
   $\varphi\iimp\psi$.

   Suppose that the left premise
   $\Gamma , \varphi \seq \psi , \Delta$ is refutable at $w$.
   Then the conclusion is also refutable at $w$ since $w\not\forces
   \Delta$ and $w \not\forces \varphi\limp\psi$.

  \item[$\impleftrulename$:] Suppose the right premise 
   $\Gamma , \varphi\iimp\psi, \psi \seq \Delta$
   is refuted by $w$. Then so is the conclusion 
   $\Gamma , \varphi\limp\psi \seq \Delta$
   since $\psi$ implies
   $\varphi\limp\psi$.
   Suppose the left premise 
   $\Gamma , \varphi\iimp\psi \seq \varphi, \Delta$
   is refuted by $w$. Since $w \not\forces
   \varphi$ and $w \forces \varphi\iimp\psi$, we must have
   $w \forces \varphi\limp\psi$. But $w \not\forces\Delta$, hence it
   refutes the conclusion.
  \end{description}
\end{proof}

% \begin{proof}
%   Again, we consider only the non-standard rules.
%   \begin{description}
%   \item[$\imprightrulename$:] Suppose the right premise
%    $\Gamma \seq \varphi\iimp\psi , \Delta $ is refuted at $w$.
%    Then so is the conclusion
%    $\Gamma \seq \varphi\limp\psi , \Delta $ 
%    since $\varphi\limp\psi$ implies 
%    $\varphi\iimp\psi$.

%    Suppose that the left premise
%    $\Gamma , \varphi \seq \psi , \Delta$ is refutable at $w$.
%    Then the conclusion is also refutable at $w$ since $w\not\forces
%    \Delta$ and $w \not\forces \varphi\limp\psi$.

%   \item[$\impleftrulename$:] Suppose the right premise 
%    $\Gamma , \varphi\iimp\psi, \psi \seq \Delta$
%    is refuted by $w$. Then so is the conclusion 
%    $\Gamma , \varphi\limp\psi \seq \Delta$
%    since $\psi$ implies
%    $\varphi\limp\psi$.
% %
%    Suppose the left premise 
%    $\Gamma , \varphi\iimp\psi \seq \varphi, \Delta$
%    is refuted by $w$. Since $w \not\forces
%    \varphi$ and $w \forces \varphi\iimp\psi$, we must have
%    $w \forces \varphi\limp\psi$. But $w \not\forces\Delta$, hence it
%    refutes the conclusion.
%   \end{description}
% \end{proof}

For a given conclusion instance of the $\steprulename$ rule, we have
already seen that conditions $\condzero$, $\condone$ and $\condtwo$
guarantee that there is at least one eventuality in the succedent, that
no termination rule is applicable, that the conclusion is
saturated, and that no eventuality in the succedent of the conclusion
is ignored.
% \begin{description}
% \item[\textbf{C2:}] $\Sigma_l$ and $\Sigma_r$ contain atomic formulae only 
% \end{description}

\begin{lemma}\label{lemma-invertibility-steprule}
  The rule
  $\steprulename$ (with $\condzero$, $\condone$ and $\condtwo$)
  is semantically invertible.
% because
 % the conclusion is refutable if 
 % some premise is refutable.
\end{lemma}

\begin{proof}
  Suppose some premise is refutable. That is,
  \begin{enumerate}
  \item 
  for some
  $1 \leq i \leq k$
  there exists a model
  $M_1 = 
    \langle 
     W_1,
     \rel_1,
     \vartheta_1
    \rangle$
%  containing a world
 and
  $w_1 \in W_1$
  such that
  $M_1, w_1
   \forces 
   \Sigma_l,
   \Theta,
   \opset{\nxt}{\Theta},
   \opset{\limp}{\Gamma},
   \varphi_i \iimp \psi_i,
   \varphi_i    
  $
  and
  $M_1, w_1
  \not\forces 
   \psi_i,
   \opsetmi{\limp}{\Delta}{i},
   \opset{}{\Phi}
  $; or
\item 
  for some
  $k+1 \leq i \leq k+n$
  there exists a model
  $M_2 =
     \langle 
     W_2,
     \rel_2,
     \vartheta_2
    \rangle$
%  containing a world
  and
  $w_2 \in W_2$
  such that
  $M_2, w_2
   \forces 
   \Sigma_l,
   \Theta,
   \opset{\nxt}{\Theta},
   \opset{\limp}{\Gamma},
   \nxt\phi_{i-k}
  $
 and $M_2, w_2
  \not\forces 
   \opset{\limp}{\Delta},
   \Phi
  $.
  \end{enumerate}

$1\leq i\leq k$:
We must show there is some model
$M$ containing a world $w_0$ such that 
$M, w_0
 \forces
 \Sigma_l,
 \opset{\nxt}{\Theta},
 \opset{\iimp}{\Gamma}
 $
and
$M, w_0
 \not\forces
 \opset{\iimp}{\Delta},
 \opset{\nxt}{\Phi},
 \Sigma_r
 $.
We do this by taking
the submodel generated by $w_1$, 
adding an extra world $w_0$ as a predecessor of
$w_1$, letting $w_0$ reach every world reachable from $w_1$, and 
setting every member of $\Sigma_l$ to be true at $w_0$.

%We have to show that there is some model
%$M$ containing a world $w_0$ such that 
%$M, w_0
% \forces
% \Sigma_l,
% \opset{\nxt}{\Theta},
% \opset{\iimp}{\Gamma}
% $
%and
%$M, w_0
% \not\forces
% \opset{\iimp}{\Delta},
% \opset{\nxt}{\Phi},
% \Sigma_r
% $.
%In each of the two cases above, we form the desired model by taking
%the submodel generated by $w_1$ or $w_2$, 
%add an extra world $w_0$ as a predecessor of
%$w_1$ or $w_2$, let $w_0$ ``reach'' every world reachable from $w_1$
%or $w_2$, and 
%put every member of $\Sigma_l$ to be true at $w_0$.
%Details are given below.

We formally define $M$ by:
   $W = \{w \in W_1 \mid w_1 \rel_1 w\}
        \cup \{w_0 , w_1\}$;
   $R =  \{(v,w) \in ~\rel_1 ~\mid v \in W, w \in W\}  
                  \cup \{(w_0 , w) \mid  w\in W\setminus\{w_0\}\}$;
%   $\rrel$ be the reflexive and transitive closure of $R$;
   for every atomic  formula $p$ and
   for every $w \in W\setminus\{w_0\}$, let
   $w\in\vartheta(p)$ iff $w\in\vartheta_1(p)$ and put
   $w_0\in\vartheta(p)$ iff $p \in \Sigma_l$.
%   $\vartheta(w, p)$ iff $\vartheta_1(w, p)$ and put
%   $\vartheta(w_0, p) = true$ iff $p \in \Sigma_l$.

   By simultaneous induction on the size of any formula
   $\xi$, it follows that for every world $w \neq w_0$ in $W$,
   we have
   $M_1, w \forces \xi$ iff
   $M, w \forces \xi$.
   % By simultaneous induction on the size of any given sub-formula
   % $\xi \in cl(\varphi_0)$ of the end-sequent $\seq \varphi_0$,
   % it follows that for every world $w \neq w_0$ in $W$,
   % $M_1, w \forces \xi$ iff
   % $M, w \forces \xi$.

   We have
   $M, w_0 \not\forces \Sigma_r$ by definition
   (since its intersection with $\Sigma_l$ is empty).
   We have
   $M, w_0 \forces \opset{\nxt}{\Theta}$ since
   $M_1, w_1 \forces \Theta$
   implies
   $M, w_1 \forces \Theta$,
   and we know that
   $w_0 \rel w_1$.
   Similarly, we have
   $M, w_0 \forces \opset{\iimp}{\Gamma}$ since
   $w_0 \rel w_1$ and
   $M_1, w_1 \forces \opset{\limp}{\Gamma}$.
   Since
   $M_1, w_1 \forces \varphi_i$ and
   $M_1, w_1 \not\forces \psi_i$, we must have
   $M, w_0 \not\forces \varphi_i \iimp \psi_i$ as desired.
   Together with
   $M_1, w_1 \not\forces \opsetmi{\limp}{\Delta}{i}$,
   we have
   $M, w_0 \not\forces \opset{\iimp}{\Delta}$.
   Finally, since
   $M_1, w_1 \not\forces \Phi$,
   we must have
   $M, w_0 \not\forces \opset{\nxt}{\Phi}$.
%   Note that
%   $w_1$ is the last refuter for 
%   $\varphi_i\limp\psi_i$
%   in $M$
%   since
%   $M_1, w_1 \forces \varphi_i\iimp\psi_i$ and hence
%   $M, w_1 \forces \varphi_i\iimp\psi_i$.
   Collecting everything together, we have
   $M, w_0 \forces
     \Sigma_l,
     \opset{\nxt}{\Theta},
     \opset{\iimp}{\Gamma}
   $ 
   and
   $M, w_0 \not\forces
    \opset{\iimp}{\Delta},
    \opset{\nxt}{\Phi},
    \Sigma_r
   $
   as desired.

The case $k+1 \leq i \leq k+n$ follows similarly.
\end{proof}

\begin{theorem}\label{thm:completeness}
%  Given an 
%  $\lcnxt$-formula
%  $\varphi_0$ 
%  (built out of atomic formulae,
%  $\top$ and
%  $\bot$ using only the connectives
%  $\land$,
%  $\lor$,
%  $\limp$),
  If the sequent
  $\seq \varphi_0$
  is not derivable using the rules
of Fig.~\ref{fig:seq-rules} according to our proof-search strategy
%conditions $\condone$ and $\condtwo$,
then $\varphi_0$ is not $\lcnxt$-valid.
\end{theorem}

\begin{proof}
  Suppose $\seq \varphi_0$ is not derivable using our
  systematic backward proof search procedure. Thus our procedure gives
  a finite tree with at least one leaf
  $\Sigma_l,  \opset{\iimp}{\Gamma}, \opset{\nxt}{\Theta} \seq \Sigma_r$ 
  obeying both $\condone$ and $\condtwo$ to which no rules are applicable.

%  We have to show that the sequent
%  $\seq \varphi_0$ is refutable by producing a model
%  $M$ and world $w$ such that $w$ refutes $\varphi$ in $M$: that is,
%  $M, w \not\forces \varphi_0$.

  Construct $M_0 = \langle W_0, \rel_0, \vartheta_0 \rangle$ as follows:
  let $W_0 = \{w_0\}$;
  let $\rel_0 = \emptyset$; and
  $w_0\in\vartheta_0(p)$ iff $p\in\Sigma_l$.
  Clearly, we have
  $M_0, w_0 \forces \Sigma_l$ by definition.
  Also,
  $M_0, w_0 \not\forces \Sigma_r$ since its intersection with 
  $\Sigma_l$ is empty by $\condone$. Every formula
  $\alpha\iimp\beta \in \opset{\iimp}{\Gamma}$
  and $\nxt\theta \in \opset{\nxt}{\Theta}$
 is vacuously true
  at $w_0$ in $M_0$ since $w_0$ has no strict successors.
 Thus the leaf sequent
  $\Sigma_l, \opset{\iimp}{\Gamma}, 
    \opset{\nxt}{\Theta} \seq 
   \Sigma_r
  $
  is refuted by $w_0$ in model $M_0$.
  The Invertibility Lemmas~\ref{lemma-invertibility-static-rules}
  and~\ref{lemma-invertibility-steprule} now
  imply that $\seq \varphi_0$ is 
  refutable in some $\lcnxt$-model.

%  If this leaf is the end-sequent we are done. Else this leaf is the
%  premise of some rule application, and the invertibility
%  Lemmas~\ref{lemma-invertibility-static-rules}
%  and~\ref{lemma-invertibility-steprule} imply that the corresponding
%  conclusion is refutable. Repeating this argument down to the end-sequent $\seq
%  \varphi_0$, produces a model $M$ and world $w$ such that 
%  $M, w \not\forces \varphi_0$.

%  If this leaf is the original end-sequent then
%  $\varphi_0$ is a single atomic formula $p$ (say), each set
%  $\Sigma_l$ and
%  $\opset{\iimp}{\Gamma}$ and
%  $\opset{\nxt}{\Theta}$ is empty, and we are done since
%  $M_0, w_0 \not\forces p$.
%
%  Else, this leaf is not the end-sequent, so the leaf is the premise
%  of some rule application giving a conclusion sequent
%  $\Gamma_1 \seq \Delta_1$. 
%  Each of our
%  rules is invertible by 
%  Lemma~\ref{lemma-invertibility-static-rules} and
%  Lemma~\ref{lemma-invertibility-steprule}.
%
%  If the rule is a static rule then the
%  world $w_0$ in the  model
%  $M_0$ refutes the conclusion
%  $\Gamma_1 \seq \Delta_1$. 
%  Else, the rule is the $\steprulename$ rule,
%  so use the construction in 
%  Lemma~\ref{lemma-invertibility-steprule}
%  to obtain a model
%  $M_1$ containing a world $w_1$
%  which refutes
%  $\Gamma_1 \seq \Delta_1$. That is,
%  $M_1, w_1 \forces \Gamma_1$ and
%  $M_1, w_1 \not\forces \Delta_1$.
%
%  Repeat this process, following the branch downward until we reach
%  the end-sequent, thus obtaining a world $w$ in a model 
%  $M$ such that $M, w \not\forces \varphi_0$, as desired.
\end{proof}

\begin{corollary}[Completeness]
\label{cor:completeness}
  If 
%the  formula
  $\varphi$ 
  is $\lcnxt$-valid then
  $\seq \varphi$
  is $\seqlcnxt$-derivable.
\end{corollary}

Cor.~\ref{cor:completeness} guarantees that any sound rule
% (such as the modus ponens rule derived earlier)
can be added to our calculus
without increasing the set of provable end-sequents,
including both forms of cut below:
\[
  \AxiomC{$\Gamma \seq \varphi, \Delta$}
  \AxiomC{$\Gamma, \varphi \seq \Delta$}
  \BinaryInfC{$\Gamma \seq \Delta$}
  \DisplayProof 
\qquad
  \AxiomC{$\Gamma, \seq \varphi, \Delta$}
  \AxiomC{$\Gamma', \varphi \seq \Delta'$}
  \BinaryInfC{$\Gamma, \Gamma' \seq \Delta, \Delta'$}
  \DisplayProof 
\]

%the context-splitting cut-rule on the right below is admissible, from
%which we can derive the context-sharing cut-rule on the left by
%putting $\Gamma_1 = \Gamma_2$ and $\Delta_1 = \Delta_2$ since our
%sequents are built from sets:
%\begin{displaymath}
%\AxiomC{$\Gamma \seq \varphi, \Delta$}
%\AxiomC{$\Gamma, \varphi \seq \Delta$}
%%\RightLabel{cut}
%\BinaryInfC{$\Gamma \seq \Delta$}
%\DisplayProof
%\qquad\quad
%\AxiomC{$\Gamma_1 \seq \varphi, \Delta_1$}
%\AxiomC{$\Gamma_2, \varphi \seq \Delta_2$}
%%\RightLabel{cut}
%\BinaryInfC{$\Gamma_1, \Gamma_2 \seq \Delta_1, \Delta_2$}
%\DisplayProof
%\end{displaymath}

%Note that 
Since all static rules are semantically invertible, any order of rule
applications for saturation suffices. Since all rules are invertible
we never need backtracking. That is, our strategy straightfowardly yields a
\emph{decision procedure}. It also tells us that $\lcnxt$, like its parent logics $\lgkm$ and $\lglc$, enjoys the finite model property:

%Note that since all static rules are semantically invertible, any order of rule
%applications suffices. That is, there is no need for backtracking
%during a saturation phase. Moreover, the rule $\steprulename$ (with
%conditions $\condzero$ and $\condtwo$) is also
%semantically invertible, hence there is no need for backtracking at
%all. That is, we can simply create one single putative derivation
%using our strategy,
%which is guaranteed to be finite.
%Moreover, we don't actually need to delay the $\steprulename$ rule
%explicitly as conditions $\condone$ and $\condtwo$ ensure that it will
%only be applicable upwards to saturated sequents which are not
%instances of a  termination rule.

%\subsection{Decidability and Finite Model Property}

\begin{theorem}\label{thm-fmp}
  If $\varphi$ is not $\lcnxt$-valid then 
  it is refutable in a rooted (finite) $\lcnxt$-model
  of length at most $l+1$ where $l$ is the length of $\varphi$.
\end{theorem}
\begin{proof}
  Suppose that $\varphi$ is not valid: that is,
  $\varphi$ is refuted by some world in some $\lcnxt$ model. By
  soundness Thm.~\ref{thm:soundness} $\seq \varphi$ is
  not derivable using our proof-search strategy. In particular, in any
  branch, there can be at most $l$ applications of the rule
  $\steprulename$ by Cor.~\ref{cor:onlylsteprules}. From such a
  branch, completeness Thm.~\ref{thm:completeness} allows us to
  construct a model $M$ and a world $w$ which refutes $\varphi$. But
  the model $M$ we constuct in the completeness proof is a rooted
  (finite) $\lcnxt$-model with at most $l+1$ worlds since the only
  rule that creates new worlds is the (transitional) $\steprulename$
  rule and there are at most $l$ such rule applications in any branch.

  % Suppose that a $\lcnxt$-formula $\varphi$ is not valid: that is,
  % $\varphi$ is refuted by some world in some $\lcnxt$ model. By the
  % soundness Theorem~\ref{thm:soundness}, the sequent $\seq \varphi$ is
  % not derivable using our proof-search strategy. In particular, in any
  % branch, there can be at most $l$ applications of the rule
  % $\steprulename$ since each eventuality which is principal in such a
  % rule application moves into the antecedent of the appropriate
  % premise and hence cannot reappear without leading to an instance of
  % $\idrulename$. From such a branch, the completeness
  % Theorem~\ref{thm:completeness} allows us to construct a model $M$
  % and a world $w$ which refutes $\varphi$. But the model $M$ we
  % constuct in the completeness proof is a rooted (finite)
  % $\lcnxt$-model with at most $l+1$ worlds since the only rule that
  % creates new worlds is the (transitional) $\steprulename$ rule and
  % there are at most $l$ such rule applications in any branch.
\end{proof}

\begin{corollary}
  $ \lcnxt$ has the finite model property.
\end{corollary}

% \begin{theorem}
%   The logic
%   $\lcnxt$ is sound and complete with respect to the class of rooted (finite)
%   $\lcnxt$-frames.
% \end{theorem}

% \begin{proof}
%   Every  rooted $\lcnxt$-frame is finite and is also a
%   $\lcnxt$ frame.  Since our logic $\lcnxt$ is
%   sound with respect to the latter class, it is also sound with
%   respect to the former class.

%   Now suppose that an $\lcnxt$-formula
%   $\varphi$ is not valid: that is, $\varphi$ is refuted by some world
%   in some model (in the latter class). By the soundness
%   Theorem~\ref{thm:soundness}, the sequent $\seq \varphi$ is not
%   derivable. Then the completeness
%   Theorem~\ref{thm:completeness} allows us to construct a model
%   $M$ and a world $w$ which refutes $\varphi$. But the model $M$ we
%   constuct in the completeness proof is a
%   rooted, finite $\lcnxt$-model.
%   Thus our logic is complete with respect to the former class as
%   claimed.
% \end{proof}

% \begin{corollary}
%   Our logic has the finite model property and is finitely
%   axiomatisable, hence it is decidable.
% \end{corollary}

%\section{Complexity}
\subsection{Complexity}

We first embed classical propositional logic into 
$\lcnxt$.
%using the standard ``double-negation translation''.

\begin{lemma}\label{lemma:double-negation}
  If
  $\varphi$
  is a formula built out of atomic formulae,
  $\top$ and
  $\bot$ using only the connectives
  $\land$,
  $\lor$,
  $\limp$, and 
  the sequent
  $\seq (\varphi \limp \bot) \limp \bot$
  is derivable, then 
  $\varphi$ is a tautology of classical propositional logic.
\end{lemma}
\begin{proof}
Any derivation in our systematic proof search
procedure ends as:
\begin{center}
\AxiomC{$\varphi \iimp \bot \seq \varphi, \bot$}
\AxiomC{$\cdots$}
\RightLabel{$\impleftrulename$}
\BinaryInfC{$\varphi \limp \bot \seq \bot$}
\AxiomC{$\cdots$}
\RightLabel{$\imprightrulename$}
\BinaryInfC{$\seq (\varphi \limp \bot) \limp \bot$}
\DisplayProof
\end{center}

Thus, the sequent 
$\varphi \iimp \bot \seq \varphi, \bot$
is derivable.

Soundness Thm.~\ref{thm:soundness} then implies that
this sequent is valid on all models. In particular, it is  valid on
the class of single-pointed models 
$M = \langle W, \rel, \vartheta \rangle$
where 
$W = \{w_0\}$ and
$\rel = \emptyset$.
The formula
$\varphi\iimp\bot$
is true at
$w_0$ vacuously since $w_0$ has no 
$\rel$-successor.
The formula
$\bot$
is not true in any model, including this one,
hence
$M, w_0 \not\forces \bot$.
Thus 
$M, w_0 \forces \varphi$.
That is, 
$\varphi$ itself is valid on all single-pointed models.
But such a model is just a valuation of classical propositional logic.
%That is, 
%$\varphi$
%is a tautology of classical propositional logic, as claimed.
\end{proof}

\begin{lemma}\label{lemma:double-negation-reverse}
  If
  $\varphi$
  is a formula built out of atomic formulae,
  $\top$ and
  $\bot$ using only the connectives
  $\land$,
  $\lor$,
  $\limp$, and 
  the sequent
  $\seq (\varphi \limp \bot) \limp \bot$
  is not derivable, then 
  $\varphi$ is not a tautology of classical propositional logic.
\end{lemma}
\begin{proof}
 Suppose
  $\seq (\varphi \limp \bot) \limp \bot$
  is not derivable. Then, by Thm.~\ref{thm:completeness}, 
  $(\varphi \limp \bot) \limp \bot$
  is not $\lcnxt$-valid. 
  Thus, there is a finite linear
  model 
  $M = \langle W, \rel, \vartheta \rangle$ with root world
  $w_0 \in W$ such that 
  $M, w_0 \not\forces (\varphi \limp \bot) \limp \bot$.
  Thus there is a world $v$ such that $w_0 \releq v$ and
  $M, v \forces \varphi \limp \bot$, which implies that every
  $\releq$-succesor of $v$, including a
  world $u$ (say) with no $R$-successors, makes $\varphi$ false. 
  But such a final world $u$ is just a valuation of classical
  propositional logic, thus there is a model of classical
  propositional logic which makes $\varphi$ false. That is,
  $\varphi$ is not a tautology of classical propositional logic.
\end{proof}

\begin{lemma}\label{lemma:np}
%  If
%  $\varphi$
%  is a $\lcnxt$-formula 
%  then 
  There is a non-deterministic algorithm to test the
  refutability (non-validity) of the sequent
  $\seq \varphi$
  in time polynomial in the length of
  $\varphi$.
\end{lemma}
\begin{proof}
  Let
  $l$ be the length of $\varphi$, and recall the definitions of 
  $sf(\varphi)$ and
  $cl(\varphi)$ given earlier.
  The number of formulae in 
  $cl(\varphi)$ is at most $2l$ and the size of each sequent our calculus
  builds is bounded by $4l^2$, since each formula of length at most
  $l$
  could appear in the
  antecedent or the succedent or both.

  Let 
  $\Gamma_1 \seq \Delta_1$,
  $\cdots$
  $,\Gamma_{k} \seq \Delta_{k}$
  be a sequence of length $k=l^2$ of sequents,
  where each sequent is built out of formulae from
  $cl(\varphi)$. 
  Check whether this sequence forms a branch of legal rule
  applications, none of which is the rule
  $\idrulename$,
  and check whether
  no rule is applicable to the sequent
  $\Gamma_{k} \seq \Delta_{k}$.
  If so, then the sequent
  $\seq \varphi$
  is refutable (non-valid).

  It remains to show that this (non-deterministic) algorithm requires
  time which is polynomial in the length $l$ of $\varphi$.  

  Every saturation phase is of length at most $l$ since each rule removes
  a connective. In any branch, there can be at most $l$ applications of the rule
  $\steprulename$ since each eventuality which is principal in such a
  rule application moves into the antecedent of the appropriate
  premise and hence cannot reappear without leading to an instance of 
  $\idrulename$. Thus every branch in any putative derivation of 
  $\seq\varphi$ is of length at most
  $k = l^2$.
  Since each sequent is of length at most 
  $4l^2$, our procedure requires at most $4l^4$ operations.
  That is, it can be done in time polynomial in the length of the
  given end-sequent.
\end{proof}

\begin{corollary}\label{cor:conp-completeness}
  The validity problem for $\lcnxt$ is coNP-complete.
\end{corollary}
\begin{proof}
  By Lem.~\ref{lemma:double-negation} we can faithfully embed the
  validity problem for classical propositional logic into $\lcnxt$,
  hence it is at least as hard as checking validity in classical
  propositional logic (coNP). By Lem.~\ref{lemma:np}, we can
  non-deterministically check non-validity of a given formula
  in time at most polynomial in its size.
\end{proof}

%%%%%%%%%%%%%%%%%%%%%%%%%%%%%%%%%%%%%%%%%%
\section{Terminating Proof Search for $\lgkm$}

This section turns to logic $\lgkm$, for which models need not be linear.
One might expect that $\lgkm$, which is conservative over $\lgint$, would require
single-conclusioned sequents only, but $\lgkm$-theorems such as the axiom
$\later\varphi\to(\psi\lor(\psi\limp\varphi))$ (see Litak~\cite{Litak:Constructive})
seem to require multiple conclusions. As such our calculus will resemble that for
$\lcnxt$. The static rules will be those of $\lcnxt$,
but the transitional rule $\steprulename$ of $\lcnxt$ is now
replaced by rules 
$\newbranchingsteprulename$ and
$\nxtrightrulename$ as shown in
Fig.~\ref{fig:step-rules-logic-KH}.

\begin{figure*}[t]
  \centering

  \begin{tabular}[c]{l@{\extracolsep{2.3em}}l}
    $\newbranchingsteprule$
    &
    $\newnxtrightrule$
    \\[15px]
    \multicolumn{2}{l}{
      where $\ddagger$ means that the following conditions hold:}
    \\[2px]
    \multicolumn{2}{l}{
      ($\condone$): $\bot \not\in \Sigma_l$ and
      $\top \not\in \Sigma_r$ and
      the conclusion is not an instance of $\idrulename$
      % $(\Sigma_l 
      % \cup \opset{\iimp}{\Gamma}
      % \cup \opset{\nxt}{\Theta})
      % \cap
      % (\Sigma_r 
      % \cup \opset{\iimp}{\Delta}
      % \cup \opset{\nxt}{\Phi})
      % = \emptyset
      %$
      }
    \\[2px]
    \multicolumn{2}{l}{
    ($\condtwo$): $\Sigma_l$ and $\Sigma_r$ contain only atomic
    formulae (\textit{i.e.} the conclusion is saturated)}
  \end{tabular}
  \caption{Transitional rules for logic $\lgkm$}
  \label{fig:step-rules-logic-KH}
\end{figure*}

The backward proof-search strategy is the same as that of
Sec.~\ref{Sec:Proof_Search}, except the transitional rule
applications now reads as below:
\begin{quote}
  \textbf{transition:} else choose a $\iimp$- or $\nxt$-formula from the
  succedent and apply $\newbranchingsteprulename$ or
  $\nxtrightrulename$, backtracking over these choices until a
  derivation is found or all choices of principal formula 
  have been exhausted.
\end{quote}

So if the given sequent is
$\seq \opset{\iimp}{\Delta}, \opset{\nxt}{\Phi},\Sigma_r$
and 
$\opset{\iimp}{\Delta}$
contains $m$ formulae
and
$\opset{\nxt}{\Phi}$
contains $n$ formulae,
then in the worst case we must explore
$m$ premise instances of 
$\newbranchingsteprulename$ and
$n$  premise instances of 
$\nxtrightrulename$.

%\subsection{Soundness}

\begin{theorem}\label{km-soundness}
  The rules
  $\newbranchingsteprulename$ 
  and
  $\nxtrightrulename$
  are \emph{sound} for the logic
  $\lgkm$. 
\end{theorem}
\begin{proof}
Suppose the conclusion of rule $\newbranchingsteprulename$ 
is refutable at world $w$
in some model $M$. Thus there is a strict $\rel$-successor $v$ of $w$
which is a last refuter for $\varphi\iimp\psi$: that is, $M, v \forces
\varphi$ and $M,v \not\forces \psi$ and
$M,v \forces \varphi\iimp\psi$. The other formulae from the antecedent
of the conclusion are also true at $v$ by truth-persistence,
and for every $\iimp$-formula
true at $w$, we also have its $\limp$-version true at $v$, and likewise for
$\nxt$-formulae.
The proof for the $\nxtrightrulename$ rule is similar.
\end{proof}

%\subsection{Cut-free Completeness}

Termination follows using the
same argument as for $\seqlcnxt$.
%That is, each saturation phase terminates since each
%static rule, read backwards, creates a strict subformula of its
%principal formula, or creates a $\iimp$-formula in the premise to
%which no static rule is (backwards) applicable. No infinite sequence
%of transitional rules can exist since each application of the
%transitional rules moves a copy of its principal formula into
%the antecedent of the premise. If this antecedent-side copy is $\nxt\psi$, then
%it persists in the antecedent of every higher sequent, and if
%$\nxt\psi$ ever reappears in the succedent of some sequent, 
%then that sequent becomes an
%instance of the rule $\idrulename$. If this antecedent-side copy is
%$\varphi\iimp\psi$, then it gets turned into $\varphi\limp\psi$ by
%every higher transitional rule, and this occurrence of
%$\varphi\limp\psi$ gets turned back into $\varphi\iimp\psi$ by the
%subsequent (higher) saturation phase.  If $\varphi\iimp\psi$ ever
%reappears in the succedent of a higher saturated sequent, then that
%sequent becomes an instance of the rule $\idrulename$. We again use
%the fact that no rule creates formulae (upwards) that do not trace back to
%a subformula of the endsequent, meaning that the number of such
%formulae is finite.
%
However the new rules are not semantically invertible, since we have to choose
a particular $\iimp$- or $\nxt$-formula from the succedent of the
conclusion and discard all others when moving to the premise, yet a
different choice may have given a derivation of the conclusion. Thus these rules
require the backtracking which is built into the new \textbf{transition}
part of our proof search strategy.

%  However, they are 
% ``existentially syntactically invertible'' in the following sense.

% \begin{lemma}\label{lemma-existential-invertibility}
% If the sequent
% $ \Sigma_l,
%  \opset{\nxt}{\Theta},
%  \opset{\iimp}{\Gamma}
%  \seq
%  \opset{\iimp}{\Delta},
%   \opset{\nxt}{\Phi},
%  \Sigma_r
% $
% obeys the $\ddagger$ conditions and is derivable then so is some premise
% $
%  \Sigma_l,
%  \Theta,
%  \opset{\nxt}{\Theta},
%  \opset{\limp}{\Gamma},
%  \varphi\iimp\psi,
%  \varphi
%  \seq
%  \psi
% $
% obtained by applying the rule
% $\newbranchingsteprulename$
% on some 
% $\varphi\iimp\psi \in  \opset{\iimp}{\Delta}$
% or so is some premise
% $
%  \Sigma_l,
%  \Theta,
%  \opset{\nxt}{\Theta},
%  \opset{\limp}{\Gamma},
%  \nxt\psi
%  \seq
%  \psi
% $
% obtained by applying the rule
% $\nxtrightrulename$
% on some 
% $\nxt\psi \in  \opset{\nxt}{\Phi}$.
% \end{lemma}

% \begin{proof}
%  Our strategy allows us to apply only the two transitional rules to a
%  sequent that obeys the $\ddagger$-conditions since these conditions
%  guarantee that no other rules are applicable to the sequent. Since
%  the sequent is derivable, then some instance of one of these rules
%  must be the source of the derivation.
% \end{proof}

\begin{lemma}\label{lemma:existential-inversion}
If a sequent $s$
obeys the $\ddagger$ conditions and every premise instance obtained by
applying the rules $\newbranchingsteprulename$ and $\nxtrightrulename$
backwards to $s$ is not derivable, then the sequent $s$ is refutable.
\end{lemma}
\begin{proof}
  We proceed by induction on the maximum number $k$ of applications of the
  transitional rules in any branch of backward proof search for $s$.

  Base case $k=0$: if $s$ obeys the $\ddagger$ conditions but contains no
  $\iimp$-formulae and contains no $\nxt$-formulae in its succedent,
  then no rule at all is applicable to $s$ and so $s$ is
  refutable as already shown in the proof of
  Thm.~\ref{thm:completeness}.

  Base case $k=1$: if $s$ obeys the $\ddagger$ conditions and the proof-search involves
  at most one application of the transitional rules in any branch,
  then each premise instance of $s$ leads upwards to at least one
  non-derivable leaf sequent to which no rule is applicable. This leaf
  is again refutable as shown in the proof of
  Theorem~\ref{thm:completeness}. The Inversion Lemmas then allow us
  to conclude that the premise instance itself is refutable since all
  rule applications in this branch must be static rules.
  Thus each
  premise instance $\pi_i$ of $s$ under the transitional rules is
  refutable in some world $w_i$ in some model $M_i$. Let $w_0$ be a
  new world and put $w_0 R w_i$ for every $w_i$ and put $w_0 R w$ for
  each $w$ which is an $\releq_i$-successor of any $w_i$ in any model
  $M_i$, and put $w_0 \in \vartheta(p)$ iff $p$ is in the antecedent
  of $s$. The new world $w_0$ makes every atomic formula in the antecedent
  of $s$ true and makes every atomic formula in the succedent of $s$
  false.  There are no conjunctions or disjunctions or $\limp$-formulae in
  $s$. 
  Every $\varphi\iimp\psi$ in the antecedent of $s$ appears in the
  antecedent of every premise instance $\pi_i$ as $\varphi\limp\psi$, so each
  $w_i$ makes $\varphi\limp\psi$ true, and hence $w_0$ makes
  $\varphi\iimp\psi$ true.
  Every $\nxt\psi$ in the antecedent of $s$ appears in the
  antecedent of every premise instance $\pi_i$ and so does $\psi$, so
  each $w_i$ makes  $\psi$ true, and hence $w_0$ makes 
  $\nxt\psi$ true.
  For every $\iimp$-formula $\varphi\iimp\psi$ in
  the succedent of $s$, the premise instance $\pi_i$ corresponding to a
  $\newbranchingsteprulename$-rule application with $\varphi\iimp\psi$
  as the principal formula will contain $\varphi$ in its antecedent 
  and contain $\psi$ in its succedent. The corresponding world $w_i$
  will make $\varphi$ true and make $\psi$
  false, meaning that $w_0$ will falsify $\varphi\iimp\psi$.
  Similarly, for every $\nxt$-formula $\nxt\psi$ in the succedent of
  $s$, the world $w_j$ obtained from a $\nxtrightrulename$-rule
  application with $\nxt\psi$ as the principal formula will falsifiy
  $\psi$, meaning that $w_0$ will falsify $\nxt\psi$. Thus $w_0$ will
  refute $s$ as claimed.

  Induction case $k+1$ for $k>0$: The induction hypothesis is that the
  lemma holds for all sequents
  $s$ that obey the $\ddagger$-conditions and whose proof-search
  involves at most $k$ applications of the transitional rules in
  any branch.

  Now suppose that $s$ obeys the $\ddagger$-conditions and the
  backward proof search for $s$ contains $k+1$ applications of the
  transitional rules. Consider the bottom-most application of the
  transitional rules (if any) along any branch ending at a premise instance
  $\pi$ of $s$. Suppose the conclusion sequent of this bottom-most
  application is $c$. This
  application falls under the induction hypothesis and so $c$ must be
  falsifiable in some model. The rules between $c$ and $\pi$ are all
  static rules, if any, and so are semantically invertible, meaning
  that the sequent $\pi$ must be falsifiable in some model. 
  Thus each
  premise instance $\pi_i$ of $s$ under the transitional rules is
  refutable in some world $w_i$ in some model $M_i$. The same
  construction as in the base case for $k=1$ suffices to deliver a
  model and a world that refutes $s$ as claimed.
\end{proof}

\begin{corollary}
 If the end-sequent $\Gamma_0\seq \Delta_0$ is not derivable using
 backward proof search according to our strategy then 
 $\Gamma_0\seq \Delta_0$ is refutable.
\end{corollary}

% \begin{corollary}\label{km-completeness}
% If 
% %the sequent 
% $\vdash \varphi_0$ is not derivable
% then the 
% formula 
% $\varphi_0$ is not $\lgkm$-valid.
% \end{corollary}

\begin{corollary}\label{km-completeness}
If $\varphi_0$ is $\lgkm$-valid then
$\vdash \varphi_0$ is $\seqkm$-derivable.
\end{corollary}

As for $\lcnxt$, our proofs yield the finite model property for $\lgkm$ as an immediate
consequence, although for $\lgkm$ this is already known~\cite{Muravitsky:Logic}.

%\begin{corollary}
% If $\seq \varphi$ is not derivable 
% then there is a finite rooted $\lgkm$-model which refutes $\varphi$
% at its root and hence $\lgkm$ has the fmp.
%\end{corollary}

\section{Related Work}

%Our logics sit at the intersection of modal logic, intuitionistic
%logic and provability logics. These have all been studied from the
%perspective of sequent calculi, algebraic semantics and Kripke
%semantics. Thus there is a huge body of related work, which we cannot
%possibly cover adequately here.

Ferrrari~et~al~\cite{DBLP:journals/jar/FerrariFF13} 
give sequent calculi for 
intuitionistic logic using
a compartment $\Theta$ in
the antecedents of their
sequents $\Theta ; \Gamma \seq \Delta$. This 
compartment contains
formulae that are not necessarily true now, but are true 
in all strict successors.
Fiorino~\cite{DBLP:journals/corr/abs-1206-4458} gives a sequent calculus using
this compartment for $\lglc$.
This yields linear depth derivations, albeit requiring a semantic check which is quadratic.
Both~\cite{DBLP:journals/jar/FerrariFF13,DBLP:journals/corr/abs-1206-4458} build in
aspects of G\"{o}del-L\"{o}b logic by allowing (sub)formulae to
cross from the succedent of the conclusion into the
compartment $\Theta$.
Our calculus differs by giving syntactic
analogues $\nxt$ and $\iimp$ for these meta-level features,
and by requiring no compartments, but
it should be possible to adapt these authors' work to 
design sequent calculi for $\lcnxt$ with linear depth derivations.

Restall~\cite{restall-subintuitionistic-logics} investigates
``subintuitionistic logics'' where each of the conditions on Kripke frames of
reflexivity, transitivity and persistence can be dropped. The logic of our novel
connective $\iimp$ can be seen as the logic bka, which lacks reflexivity, but
has the additional conditions of linearity and converse well-foundedness, which Restall
does not consider. 
The models studied by Restall all require a root world, and thus
they disallow sequences $\cdots x_3 R x_2 R x_1$ 
which are permitted by $\lcnxt$-models.
Ishigaki and Kikuchi~\cite{DBLP:conf/tableaux/IshigakiK07} 
give ``tree-sequent'' calculi for the first-order versions of
some of these subintuitionistic logics. 
%It is easy to see that 
%Their
%tree-sequents are labelled sequents where the underlying frame must be
%a tree.
Thus ``tree-sequent'' calculi
for $\lgkm$ and $\lcnxt$ are possible, but our calculi require no labels.

Labelled sequent calculi for $\lgkm$
and $\lcnxt$ are possible by extending the work of 
Dyckhoff and Negri~\cite{DBLP:journals/aml/DyckhoffN12} but
termination proofs and complexity results
for labelled calculi are significantly harder than
our proofs. 
%
% Dyckhoff and Negri~\cite{DBLP:journals/aml/DyckhoffN12} give
% labelled sequent calculi for numerous super-intuitionistic logics. It
% should be possible to obtain a labelled calculus for $\lcnxt$ but
% termination proofs for labelled calculi are significantly harder than
% the one we give. It is not obvious how to obtain complexity
% results from labelled calculi.

Garg et al~\cite{DBLP:conf/lics/GargGN12} give labelled
sequent
calculi for intuitionistic modal logics and general conditions on
decidability. Their method relies on a first-order
characterisation of the underlying Kripke relations, but
converse well-foundedness is not first-order definable.
Labelled calculi can handle converse well-founded frames 
by allowing formulae to ``cross'' sides as in our
calculus, but it is not clear whether the method of 
Garg et al~\cite{DBLP:conf/lics/GargGN12} then applies.

%Lahav~\cite{DBLP:conf/lics/Lahav13}
%gives a general method to turn
%certain first-order conditions into cut-free complete hyper-sequent
%calculi for the corresponding propositional modal logics. These
%results cannot capture converse well-foundedness since it is not
%first-order definable. He also does not consider intuitionistic
%logics.

% Ciabattoni \textit{et. al.}~\cite{DBLP:conf/csl/CiabattoniST09} give
% general methods for obtaining cut-free complete hyper-sequent
% calculi for many intuitionistic substructural logics. Our
% connective  $\iimp$ is ``sub-intuitionistic'', so their results do not
% apply directly to our logics.

% $\lglc$ has also been studied as a fuzzy logic, and hyper-sequent
% calculi for it have been given by ... This allows the sequents to be
% single-succedent while allowing us room for a structural connective
% for $\lor$ in the succedent of sequents.

% Tiu has given a method for adding quantifiers to
% $\lglc$ which captures XXX domains.

% The logic
% $\lglc$ has
% interpolation~\cite[page~455]{chagrov-zakharyaschev-modal-logic}.

Our complexity results follow directly from our calculi; a possible alternative
may be to adapt the polynomial encoding
% validity-preserving polynomial translation
of $\lglc$ into
classical satisfiability~\cite{chagrov-zakharyaschev-modal-logic}.

%%%%%%%%%%%%%%%%%%%%%%%%%%%%%%%%%%%%%%%%%%
\section{Conclusion}

We have seen that 
the internal \emph{propositional} logic of the topos
of trees is $\lcnxt$. Indeed it may be tempting to think that $\lcnxt$
is just $\lglc$, as both are sound and complete with respect to the
class of finite sequences of reflexive points, but note that we cannot
express the modality $\nxt$ in terms of the connectives of $\lglc$.

Linear frames seem concordant with the \emph{step-indexing} applications of later,
based as they are on induction on the natural numbers rather than any
branching structure, but seem less natural from a \emph{types} point of
view, which tend to build on intuitionistic logic. For a possible type-theoretic
intepretation of linearity see Hirai's $\lambda$-calculus for $\lglc$ with applications to `waitfree' computation~\cite{Hirai:Lambda}. More broadly our work provides a
proof-theoretical basis for future research into computational aspects of intuitionistic
G\"{o}del-L\"{o}b provability logic.

The topos of trees, which generalises some previous models, has itself been
generalised as a model of guarded recursion in several ways~\cite{Birkedal:First,Birkedal:Intensional,Milius:Guard}. These categories do not all
correspond to $\lcnxt$; some clearly fail to be linear. The logical content of these
general settings may also be worthy of study.

The most immediate application of our proof search algorithm may be to provide
automation for program logics that use
later~\cite{Hobor:Oracle,Bengtson:Verifying,Clouston:Programming}. Support for
a richer class of connectives, such as first and higher order quantifiers, would be
desirable. We in particular note the `backwards looking box' used
by Bizjak and Birkedal~\cite{Bizjak:Model} in sheaves over the first uncountable
ordinal $\omega_1$, and subsequently in the topos of trees by Clouston et
al~\cite{Clouston:Programming} to reason about
coinductive types.

%\subsection*{Acknowledgments}
\subsubsection*{Acknowledgments}
%\noindent\textbf{Acknowledgments} 
We gratefully acknowledge helpful discussions with Lars Birkedal,
Stephan\'e Demri, Tadeusz Litak, and Jimmy Thomson, and the comments
of the reviewers of this and a previous unsuccessful submission.

\bibliographystyle{splncs03}
\bibliography{original-blbl}

\end{document}